\documentclass[10pt, conference, letterpaper]{IEEEtran}

\usepackage{bbm}
\usepackage{amssymb}
\usepackage{amsthm}
\usepackage{mathtools}
\usepackage{graphicx}
\usepackage{epstopdf}
\DeclareGraphicsExtensions{.eps}
\usepackage{algorithm, amsmath}
\usepackage{algpseudocode}
\usepackage{float}
\usepackage{subcaption}
\usepackage{etoolbox}
\usepackage[T1]{fontenc}

\newcommand{\V}{\mathcal{V}}
\newcommand{\E}{\mathcal{E}}
\allowdisplaybreaks

\usepackage{tikz}
\usetikzlibrary{arrows.meta}
\tikzset{>={Latex[width=2mm,length=2mm]}}

\tikzstyle{vertex}=[circle, draw]
\newtheorem{theorem}{Theorem}
\newtheorem{lemma}{Lemma}

\tikzstyle{stuff_fill}=[vertex, style=green, fill=black!10]
\usetikzlibrary{calc}

\usepackage{cleveref}
\IEEEoverridecommandlockouts
\begin{document}

\title{Shortest Path and Maximum Flow Problems Under Service Function Chaining Constraints}

\author{Gamal Sallam, Gagan R. Gupta, Bin Li, and Bo Ji \thanks{This work was supported in part by the NSF under Grants CNS-1651947 and CNS-1717108. 
		
		Gamal Sallam (tug43066@temple.edu) and Bo Ji (boji@temple.edu)  are with the Department of
		Computer and Information Sciences, Temple University, Philadelphia,	PA,  Gagan R. Gupta (gagan.gupta@iitdalumni.com) is	with AT\&T Labs, and Bin Li (binli@uri.edu) is with the Department of
		Electrical, Computer and Biomedical Engineering, University of Rhode Island, Kingston, Rhode Island.
}}

% make the title area
\maketitle

\begin{abstract} 
With the advent of Network Function Virtualization (NFV), Physical Network Functions (PNFs) are gradually being replaced by Virtual Network Functions (VNFs) that are hosted on general purpose servers. Depending on the call flows for specific services, the packets need to pass through an ordered set of network functions (physical or virtual) called Service Function Chains (SFC) before reaching the destination. Conceivably for the next few years during this transition, these networks would have a mix of PNFs and VNFs, which brings an interesting mix of network problems that are studied in this paper: (1) How to find an SFC-constrained shortest path between any pair of nodes? (2) What is the achievable SFC-constrained maximum flow? (3) How to place the VNFs such that the cost (the number of nodes to be virtualized) is minimized, while the maximum flow of the original network can still be achieved even under the SFC constraint? In this work, we will try to address such emerging questions. First, for the SFC-constrained shortest path problem, we propose a transformation of the network graph to minimize the computational complexity of subsequent applications of any shortest path algorithm. Second, we formulate the SFC-constrained maximum flow problem as a fractional multicommodity flow problem, and develop a combinatorial algorithm for a special case of practical interest. Third, we prove that the VNFs placement problem is NP-hard and present an alternative Integer Linear Programming (ILP) formulation. Finally, we conduct simulations to elucidate our theoretical results.
\end{abstract}

\IEEEpeerreviewmaketitle

\section{Introduction}

Major service providers in the communications industry across the globe are transforming their technology, operations and business models to harness the benefits of Network Function Virtualization (NFV). Stated simply, NFV involves replacing the  Physical Network Functions (PNFs) running on commodity hardware with software modules called Virtual Network Functions (VNFs) that are hosted on general purpose servers \cite{Introduction2012}. Each server can host multiple VNFs, while each network function can have multiple instances running at different physical locations. Hybrid networks comprising the VNFs and legacy PNFs will be the norm for the next decade \cite{Sherry2012, Amdocs_whitepaper}. Even in a hybrid network, a lot of benefits can be harnessed. For instance, flows can be processed by different functions at one node and functions can be flexibly added and removed. This opens up an interesting mix of network problems that are studied in this paper.

Service function chaining (SFC) is the ability to specify a set of network functions as well as their execution order for each flow \cite{Bhamare2016}. 
An example is provided in Fig. \ref{fig:scenario} in which different network functions are supported at different locations. Assume that we have a flow from $v_1$ to $v_6$, with the following SFC constraint: ($v_1$, Firewall (FW), wide area network (WAN) optimizer, $v_6$). Different paths that satisfy the SFC constraint are available for this flow such as $(v_1,  v_2,  v_4, v_3,  v_6)$, or $(v_1,  v_4,  v_5,  v_6)$. Which of these paths to choose depends on the load at each instance and the total congestion along each path. Moreover, satisfying the SFC constraint may reduce the maximum flow that can be sent from $v_1$ to $v_6$, because some paths (e.g., path $(v_1,  v_2,  v_3,  v_6)$) do not satisfy the SFC constraint. To achieve the original maximum flow, the decision of where to place the network functions should be made carefully to ensure that any fraction of the maximum flow passes through the required network functions.

\begin{figure}
	\centering
	
	\begin{tikzpicture}[transform shape]
	\node[align=left] at (6.2,2) {FW: Firewall\\ DPI: Deep packet inspection\\ WAN: Wide area network};
	
	\node[vertex](v_1) at (0, 0) {$ v_1 $};
	
	\node[vertex, label=above:{\{FW, DPI\}}](v_2) at (2, 1) {$ v_2 $};
	\node[vertex](v_3) at (4, 1){$ v_3 $};
	\node[vertex ](v_4) at (2, -1){$ v_4 $};
	\node[align=left] at (1.5,-1.64) {\{Proxy,  WAN optimizer\}};
	\node[vertex](v_5) at (4, -1) {$ v_5 $};
	\node[align=left] at (5.5,-1.64) {\{FW, WAN optimizer\}};
	
	\node[vertex](v_6) at (6, 0) {$ v_6 $};
	
	\begin{scope}[every path/.style={-}, every node/.style={inner sep=1pt}]
	\path (v_1) edge [ anchor= south] node {} (v_2);
	\path (v_1) edge [ anchor= south] node {} (v_4);
	\path (v_2) edge [ anchor= east] node {} (v_3);
	\path (v_2) edge [ anchor= south] node {} (v_4);
	\path (v_2) edge [ anchor= south] node {} (v_5);
	\path (v_4) edge [ anchor= south] node {} (v_5);
	\path (v_3) edge [ anchor= south] node {} (v_5);
	\path (v_3) edge [ anchor= south] node {} (v_6);
	\path (v_5) edge [ anchor= south] node {} (v_6);
	\path (v_4) edge [ anchor= south] node {} (v_3);
	
	\end{scope} 
	\end{tikzpicture}
	\caption{A network with network functions at different locations.}
	\label{fig:scenario}
\end{figure}
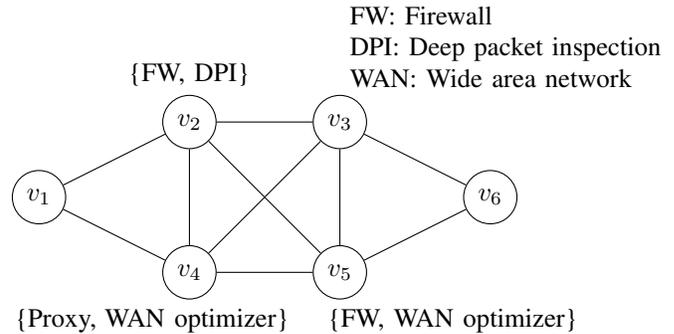

The first problem we consider in this work is, how to efficiently compute a shortest path that satisfies a given SFC constraint? Clearly, classic shortest path algorithms, such as Dijkstra's algorithm, may give a path that no longer satisfies the SFC constraint. To that end, we propose an algorithm for computing an SFC-constrained shortest path based on transforming the network graph $G$ into a new graph $\bar{G}$. Then, any shortest path found for the new graph $\bar{G}$ can be mapped to an SFC-constrained shortest path over the original graph $G$. Further, we develop a pruning algorithm that can greatly reduce the size of $\bar{G}$. As long as the network topology and SFCs remain the same, $\bar{G}$ can be repeatedly used to compute SFC-constrained shortest paths as the network costs change during the course of time. In some cases, the order of some of the network functions can be transposed, and we model that by allowing a set of valid SFCs. Moreover, it is worth noting that our proposed shortest path algorithm can also be integrated into a unified throughput-optimal routing framework \cite{OptimalControl} to achieve throughput optimality for unicast flows with SFC constraints.

Then, we consider another classic problem, the maximum flow problem, again under the SFC constraint. We call this problem the SFC-constrained maximum flow (SFC-MF) problem. The objective is to find the maximum feasible flow from a source to a destination that satisfies a given SFC constraint. We formulate the SFC-MF problem as a fractional multicommodity flow problem, which can be solved using a Linear Programming (LP) solver or approximation algorithms \cite{Cao2014}. An interesting use case is when a service provider needs to virtualize a particular network function in its network during an early stage of NFV deployment. We consider the problem of computing the maximum flow with the constraint that all packets must pass through this new VNF. We propose an elegant combinatorial algorithm for this case based on the Ford-Fulkerson algorithm \cite{Ford}.

Note that the value of the SFC-MF is not necessarily equal to that of the original maximum flow, which apparently depends on the placement of the network functions. Hence, an important question is how to place a set of VNF instances such that the original maximum flow (without the SFC constraints) can still be achieved, while the placement cost is minimized. To minimize the total operational expenses of adding commodity servers in the network to support VNFs, we aim to minimize the number of network nodes where these functions will be placed.  We first prove that this problem is NP-hard based on a reduction from the classic set-cover problem. Then, we present an alternative Integer Linear Programming (ILP) formulation that is shown to solve a large instance (e.g., a network with 100 nodes) in a few minutes. We observe via simulations that for random graphs, the maximum flow can be achieved by placing the VNFs on a small number of nodes even when the graph is large. This indicates that the operators may be able to introduce VNFs in their networks at a low starting cost without impacting the capacity, i.e., the amount of flow that can be sent.

The rest of the paper is organized as follows. In Section \ref{sec:Related_Work}, we position our paper in comparison to prior art. In Section \ref{sec:system_model}, we introduce the system model.  We investigate the SFC-constrained shortest path problem in Section \ref{sec:shortest_path}. Then, in Section \ref{sec:SFC-MF}, we describe the SFC-constrained maximum flow and present a combinatorial solution for a special case  of practical interest.  In Section \ref{sec:VSF_placement}, we focus on the problem of VNFs placement. Finally, we present the simulation results in Section \ref{sec:numerical_analysis} and conclude the paper in Section \ref{sec:conclusion}.

\section{Related Work}
\label{sec:Related_Work}
\textbf{SFC-constrained Shortest Path.} The problem of SFC-constrained shortest path has been considered in \cite{Dwaraki2016,Cao2014}. Specifically, in \cite{Dwaraki2016}, a layered graph with $r+1$ layers is constructed where each layer is a replication of the original graph and $r$ is the number of network functions in a given SFC constraint. Then, an SFC-constrained shortest path over the original network graph can be found by applying any shortest path algorithm over the layered graph. However, the constructed layered graph has a large size. Specifically, the number of nodes and edges increases by at least a factor of $r+1$ compared to the original graph. In \cite{Cao2014}, another approach is proposed, which requires the computation of the shortest paths between all node-pairs in order to find a specific SFC-constrained shortest path. Restricting a path to be a simple path with multiple must-stop nodes, without any order requirements, is NP-Complete and in \cite{Vardhan2009}, a heuristic is proposed for that. In our proposed approach, we will construct a new graph that has a small size, needs to be constructed only once, and does not require any further changes after each shortest path computation.

\textbf{SFC-constrained Maximum Flow.} The work of \cite{Charikar2015} formulates an SFC-constrained Maximum Flow problem as a multicommodity maximum flow problem. Hence, their problem is an LP and can be solved using any LP solver or can be approximated using a multiplicative weight update method \cite{Arora2012}. Another approximation algorithm is presented in \cite{Cao2014} to decide if a given set of flows with SFC constraints can be supported. In contrast, in this paper we are interested in combinatorial algorithms that give exact solutions.

\textbf{Placement of VNFs}. The problem of VNFs placement with different objectives has been studied in the past few years. In \cite{Sang}, the authors focus on the placement of VNF instances that satisfies the demands of flows with given routes. A similar problem is investigated in \cite{poularakis2017one} but for gradually upgrading some nodes to have Software Defined Networking (SDN) capabilities. Specifically, the authors consider how to select a set of nodes to upgrade to SDN such that the flow that passes through at least one SDN node is maximized, where each flow has a predetermined path. The work of \cite{Feng2017} considers a joint problem of VNFs placement and flow routing to minimize the total amount of resources used by the flows, while in \cite{Feng2016}, the objective is to ensure that the underlying network is stable. In \cite{Bhamare2017}, the authors consider the problem of VNFs placement for minimizing both the end user delay and deployment cost. Considering a similar placement problem, the work of \cite{Li2016a} aims to maximize the number of admitted requests and provides a soft real-time guarantee for each admitted request. In \cite{Ma2017}, the authors consider how to minimize the overall traffic volume for a given flow given that the traffic volume may change (increase or decrease) after being processed by some network functions. Different from all prior works, in the new placement problem we will consider, the objective is to minimize the number of nodes that need to be virtualized such that, the original maximum flow can be achieved under a given SFC constraint.

\section{System model}
\label{sec:system_model}
We consider a network that is represented by a graph $G=(\V, \E)$, where $\V$ denotes the set of vertices and $\E$ denotes the set of edges. We will first consider directed graph for the shortest path problem and placement problem, and then, consider undirected graph for the maximum flow problem. We use $\phi_i$ to denote network function $i$, and use $\Phi_{v_k}$ to denote the set of network functions supported at node $v_k$. The network functions are either physical devices attached to network nodes or virtual network functions at servers or datacenters attached to network nodes. In the case of dataceners, we assume that the capacity can be enlarged as needed. The same network function may have multiple instances at different vertices. We consider a flow that is required to satisfy an SFC constraint represented as: $(v_s, \phi_1, ..., \phi_r, v_d)$, where $v_s$ and $v_d$ are the source and destination, respectively, and $(\phi_1, ..., \phi_r)$ denotes a sequence of network functions by which all the packets of the flow need to be processed before reaching $v_d$.  

A path $p$ is denoted by $p=(e_1, ...,e_{|p|})$, where $e_i$ is the $i^{th}$ hop edge of path $p$ and $|p|$ is the length of path $p$. Sometimes, we refer to a path by the nodes along the path, i.e., $p=(v_1, v_2, \dots, v_n)$, where $e_i=(v_i, v_{i+1})$. By slightly abusing the notations, we also use $e \in p$ to denote an edge of path $p$. A path is called \emph{admissible} for a flow if it satisfies the requirements of the flow specified by the SFC constraint. To ensure the order imposed by an SFC, a packet may need to visit the same vertex more than once before it reaches the destination.

\section{SFC-constrained Shortest Path}

\label{sec:shortest_path}
In this section, we focus on the shortest path problem with a given SFC constraint. Specifically, among all the admissible paths we want to find the one that has the minimum cost, which could be the least congestion level (i.e., smallest delay). Note that the classic shortest path algorithms cannot be directly applied to produce a shortest path, because the generated path may not satisfy the given SFC constraint. For instance, consider the network presented in Fig. \ref{fig:networkGraph1}. We have two network functions, $\phi_1$ at nodes $ v_2 $ and $ v_4 $, and $\phi_2$ at nodes $ v_2 $ and $ v_3 $. Consider an $\text{SFC} = ( v_1, \phi_1, \phi_2, v_5)$. Any conventional shortest path algorithm will return a path $(v_1,  v_3,  v_5)$, which does not satisfy the imposed SFC constraint, i.e., the flow needs to be processed by $\phi_1$ first before it is processed by $\phi_2$. Another solution is to find a shortest path from the source $ v_1 $ to the first network function $\phi_1$, which is $(v_1,  v_2)$, then from $ v_2 $ we find the shortest path to $\phi_2$, which is $ v_2 $ itself, then from $ v_2 $ to the destination. This solution results in a path of $(v_1,  v_2,  v_5)$, which has a cost of 8. However, it can be verified that path $(v_1,  v_3,  v_4,  v_3,  v_5)$ satisfies the SFC constraint and has the minimum cost of 6. From this simple example, we can observe that it is non-trivial to find a path with minimum cost while satisfying the given SFC constraint. 

We propose a novel solution by cleverly transforming the network graph, $G$, to a new graph, $\bar{G}$, in which the shortest path will be computed and mapped to a path in $G$. We describe this algorithm in Algorithm \ref{alg:SP_SFC} and explain its operations in detail as follows.
\begin{figure}
	\centering
	
	\begin{tikzpicture}[transform shape]
	
	\node[vertex](v_1) at (-1, 1) {$ v_1 $};

	\node[vertex, label=below:{ $\Phi_{v_2}=\{\phi_1, \phi_2\}$}](v_2) at (2, 1) {$ v_2 $};
	\node[vertex, label=below:{$\Phi_{v_3}=\{\phi_2\}$}](v_3) at (2, -1){$ v_3 $};
	\node[vertex, label=below:{$\Phi_{v_4}=\{\phi_1\}$}](v_4) at (-1, -1){$ v_4 $};
	\node[vertex](v_5) at (4, 0) {$ v_5 $};
	
	\begin{scope}[every path/.style={->}, every node/.style={inner sep=1pt}]
	\path (v_1) edge [ anchor= south] node {$3$} (v_2);
	\path (v_1) edge [ anchor= south] node {$1$} (v_3);
	\path (v_1) edge [ anchor= east] node {$5$} (v_4);
	\path (v_2) edge [ anchor= south] node {$5$} (v_5);
	\path (v_3) edge [ anchor= south] node {$1$} (v_5);
	\path (v_4) edge [ anchor= south] node {$1$} (v_3);
	\path (v_3) edge [ anchor= south, bend left] node {$3$} (v_4);
	
	\end{scope} 
	\end{tikzpicture}
	\caption{Graph representation of the network, where each node has zero or more network functions, shown below each node.}
	\label{fig:networkGraph1}
\end{figure}
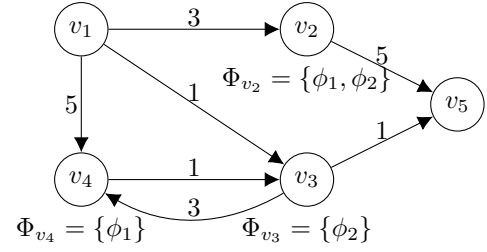
\subsection{Constructing ${\bar{G}}$ }
\label{subsection:G_bar_construction}
\begin{algorithm}[t]
	\caption{ SFC-constrained shortest path algorithm}
	\label{alg:SP_SFC}
	\begin{algorithmic}[1]
		\State \textbf{Require}: $G= (\V, \E)$,  and $\text{SFC} = (v_s, \phi_1, \dots, \phi_r, v_d)$.
		\State \textbf{Output}: A shortest path from $v_s$ to $v_d$ that satisfies the SFC constraint.
		\State \textbf{// Initial Construction of ${\bar{G}}=(\bar{\V}, \bar{\E})$}
		\For {each $v_k$ in $\V$}
		\State \parbox[t]{\dimexpr\linewidth-\algorithmicindent} {Create $r+1$ virtual vertices $v_k^0, ..., v_k^r $ and add \newline them to $\bar{\V}$.\strut}
		\EndFor
		\For {each edge $(v_l, v_k)$ in $\E$}
		\For {each $v_l^i$ in $\bar{\V}_{v_l}$}
		\State {Pick $v_k^{j^*}$ in $\bar{\V}_{v_k}$ such that}	
		\begin{equation}					
		j^* = \max \{j: \{\phi_{i+1}, \dots, \phi_{j}\} \subseteq \Phi_{v_k}, j \geq i \}
		\label{eq:newEdges}
		\end{equation}
		\State {Add edge $(v_l^i, v_k^{j^*})$ to $\bar{\E}$ .}
		\EndFor
		\EndFor
		\State \textbf{ // Pruning ${\bar{G}}$}
		\While{there is a vertex $v$ in $\bar{\V}$ with no incoming edge(s) or outgoing edge(s), except the source and destination}
		\State {remove $v$ and the edges that connect to $v$.}
		\EndWhile
		\Statex
		\State \textbf{ // Computing a shortest path }
		\State {Define $t$ as the first network functions of the SFC that are available at the source}
		\State {Map the source $v_s \in \V$ to $v_s^t\in \bar{\V}$}
		\State {Map the destination $v_d \in \V$ to $v_d^r \in \bar{\V}$}
		\State {Use any shortest path algorithm (e.g., Dijkstra's algorithm) to compute a shortest path ${\bar{p}}$ from $v_s^t$ to $v_d^r$ in $\bar{{G}}$.}
		\State {Map ${\bar{p}}=(v_s^t,  \dots,  v_k^{t+1},  \dots,  v_d^r)$ to path ${p} = (v_s,  \dots,  v_k,  \dots,  v_d)$ in ${G}$}.
	\end{algorithmic}
\end{algorithm}

\subsubsection{Initial ${\bar{G}}$}

Given a network represented by a graph $G=(\V,\E)$ and and SFC = $(v_s, \phi_1, ..., \phi_r, v_d)$, we will construct a new graph ${\bar{G}}= (\bar{\V}, \bar{\E})$.  The vertices $\bar{\V}$ are as follows. For each $v_k$ in $\V$, we create $r+1$ virtual vertices $v_k^i$ (some of which will be removed later), where $i=0, ..., r$. The idea is to ensure that each virtual vertex $v_k^i$ has the following reachability property: it is reachable from the source only if the path from the source to $v_k^i$ satisfies the partial service function chain $(\phi_1, ..., \phi_i)$. To do that, the edges in $\bar{\E}$ are established as follows. First, we use $\bar{\V}_{v_k}$ to denote the set of vertices in ${\bar{G}}$ that corresponds to vertex $v_k$ in $G$. We construct edges in ${\bar{G}}$ as follows. For each edge $(v_l, v_k)$ in $G$, an edge is established between each pair of $v_l^i \in \bar{\V}_{v_l}$ and $v_k^{j^*} \in \bar{\V}_{v_k}$, where $j^* \geq i$ is the highest index for which the set $(\phi_{i+1}, \dots, \phi_{j^*})$ is a subset of $\Phi_{v_k}$, as in Eq. \eqref{eq:newEdges}. That means the selected vertex $v_k^{j^*}$ is either a vertex that represents the same network functions as $v_l^i$, i.e., $j^* = i$, or $v_k^{j^*}$ represents more network functions than $v_l^i$, but this difference in the network functions is supported by node $v_k$. For instance, $v_l^1$ can be connected to $v_k^2$ if the network function $\phi_2$ is supported by node $v_k$, otherwise, we connect it to vertex $v_k^1$. The cost of edge $(v_l^i, v_k^{j^*})$ is set to the same cost of edge $(v_l, v_k)$. 

\subsubsection{Pruned ${\bar{G}}$}
Note that some vertices in ${\bar{G}}$ only have outgoing edges and do not have any incoming edges. Such vertices, except for the source, can be removed from ${\bar{G}}$ because they will not contribute to any path from the source. For a similar reason, we remove all the vertices that only have incoming edges and do not have any outgoing edges. This procedure will create new vertices that do not have any incoming edges or that do not have any outgoing edges. Hence, we repeat this procedure until no such vertices exist.

\subsubsection{Computing a Shortest Path}
\label{sec:shortest_path_approach}

After constructing graph ${\bar{G}}$, we find the shortest path from a source $v_s$ in $\V$  to a destination $v_d$ in $\V$ with an SFC constraint of length $r$ as follows. Assume that the first $t$ network functions of the SFC constraint are available at the source. We run any shortest path algorithm to find a path from $v_s^t$ to $v_d^r$ in ${\bar{G}}$. The obtained path is mapped to a path in ${G}$ by replacing each vertex with its corresponding vertex in ${G}$. In some cases, the SFC constraint has some flexibility, which means that some functions can be implemented in any order. So, we utilize that by constructing a set of valid SFCs and for each SFC, we construct a graph $\bar{G}$. Then, we find the shortest path in each constructed graph $\bar{G}$ and pick the shortest path among all of them. However, if we have a fully flexible SFC constraint, then the number of valid SFCs for $r$ functions will be $r!$.

\subsection{Detailed Example}

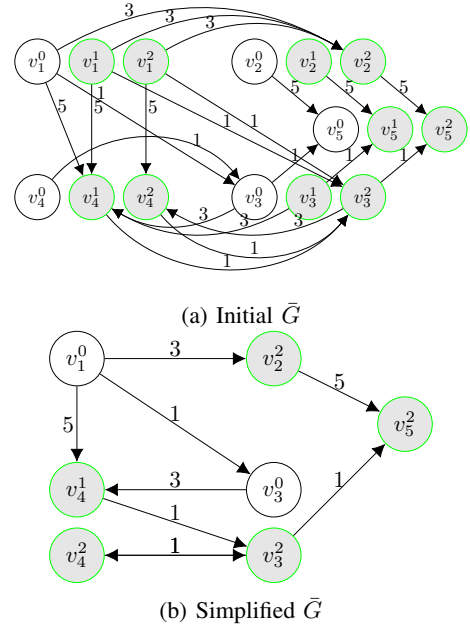
\begin{figure}
	\centering
	\begin{subfigure}{0.7\linewidth}
		\centering

		\resizebox{\linewidth}{!}{
			\begin{tikzpicture}[transform shape]
			\node[vertex](v_1^0) at (-3, 1.5) {$v_1^0$};
			\node[stuff_fill](v_1^1) at (-2, 1.5) {\textcolor{black}{$ v_1^1 $}};
			\node[stuff_fill](v_1^2) at (-1, 1.5) {\textcolor{black}{$ v_1^2 $}};
			\node[vertex](v_2^0) at (1, 1.5) {$v_2^0$};
			\node[stuff_fill](v_2^1) at (2, 1.5) {\textcolor{black}{$ v_2^1 $}};
			\node[stuff_fill](v_2^2) at (3, 1.5) {\textcolor{black}{$ v_2^2 $}};
			\node[vertex](v_4^0) at (-3, -1) {$ v_4^0 $};
			\node[stuff_fill](v_4^1) at (-2, -1) {\textcolor{black}{$ v_4^1 $}};
			\node[stuff_fill](v_4^2) at (-1, -1) {\textcolor{black}{$ v_4^2 $}};
			\node[vertex](v_3^0) at (1, -1) {$ v_3^0 $};
			\node[stuff_fill](v_3^1) at (2, -1) {\textcolor{black}{$ v_3^1 $}};
			\node[stuff_fill](v_3^2) at (3, -1) {\textcolor{black}{$ v_3^2 $}};
			\node[vertex](v_5^0) at (2.5, 0.25) {$ v_5^0 $};
			\node[stuff_fill](v_5^1) at (3.5, 0.25) {\textcolor{black}{$ v_5^1 $}};
			\node[stuff_fill](v_5^2) at (4.5, 0.25) {\textcolor{black}{$ v_5^2 $}};
			
			\begin{scope}[every path/.style={->}, every node/.style={inner sep=1pt}]
			\path (v_1^0) edge [ anchor= west, near start] node {$5$} (v_4^1);
			\path (v_1^1) edge [ anchor= west, , near start] node {$5$} (v_4^1);
			\path (v_1^2) edge [ anchor= west, near start] node {$5$} (v_4^2);
			\path (v_1^0) edge [ anchor= south, bend left, near start] node {$3$} (v_2^2);
			\path (v_1^1) edge [ anchor= south, bend left, near start] node {$3$} (v_2^2);
			\path (v_1^2) edge [ anchor= south, bend left, near start] node {$3$} (v_2^2);
			\path (v_4^0) edge [ anchor= south, bend left=50, near end] node {$1$} (v_3^0);
			\path (v_4^1) edge [ anchor= south, bend right=50] node {$1$} (v_3^2);
			\path (v_4^2) edge [ anchor= south, bend right=50] node {$1$} (v_3^2);
			\path (v_1^0) edge [ anchor= south, near start] node {$1$} (v_3^0);
			\path (v_1^1) edge [ anchor= south] node {$1$} (v_3^2);
			\path (v_1^2) edge [ anchor= south] node {$1$} (v_3^2);
			\path (v_3^0) edge [ anchor= south] node {$1$} (v_5^0);
			\path (v_3^1) edge [ anchor= south] node {$1$} (v_5^1);
			\path (v_3^2) edge [ anchor= south] node {$1$} (v_5^2);
			\path (v_2^2) edge [ anchor= south] node {$5$} (v_5^2);
			\path (v_2^0) edge [ anchor= south] node {$5$} (v_5^0);
			\path (v_2^1) edge [ anchor= south] node {$5$} (v_5^1);
			\path (v_3^0) edge [ anchor= south , bend left, near start] node {$3$} (v_4^1);
			\path (v_3^1) edge [ anchor= south, bend left, near start] node {$3$} (v_4^1);
			\path (v_3^2) edge [ anchor= south, bend left, near start] node {$3$} (v_4^2);
			
			\end{scope} 
			\end{tikzpicture}
		}
		\caption{Initial $\bar{G}$}
		\label{fig:newNetworkGraph1}
	\end{subfigure}
	
	\begin{subfigure}{0.6\linewidth}
		\centering
		\resizebox{\linewidth}{!}{
			\begin{tikzpicture}[transform shape]
			\node[vertex](v_1^0) at (1, 1) {$ v_1^0 $};
			\node[stuff_fill](v_2^2) at (4, 1) {\textcolor{black}{$ v_2^2 $}};
			\node[stuff_fill](v_4^1) at (1, -1) {\textcolor{black}{$ v_4^1 $}};
			\node[stuff_fill](v_4^2) at (1, -2) {\textcolor{black}{$ v_4^2 $}};
			\node[vertex](v_3^0) at (4, -1) {$ v_3^0 $};
			\node[stuff_fill](v_3^2) at (4, -2) {\textcolor{black}{$ v_3^2 $}};
			\node[stuff_fill](v_5^2) at (6, 0) {\textcolor{black}{$ v_5^2 $}};
			
			\begin{scope}[every path/.style={->}, every node/.style={inner sep=1pt}]
			\path (v_1^0) edge [ anchor= east] node {$5$} (v_4^1);
			\path (v_1^0) edge [ anchor= south] node {$1$} (v_3^0);
			\path (v_3^0) edge [ anchor= south] node {$3$} (v_4^1);
			\path (v_1^0) edge [ anchor= south] node {$3$} (v_2^2);
			\path (v_4^1) edge [ anchor= south] node {$1$} (v_3^2);
			\path (v_3^2) edge [ anchor= south] node {$1$} (v_5^2);
			\path (v_3^2) edge [ anchor= south] node {$1$} (v_4^2);
			\path (v_4^2) edge [ anchor= south] node {$1$} (v_3^2);
			\path (v_2^2) edge [ anchor= south] node {$5$} (v_5^2);
			\end{scope}
			\end{tikzpicture}
		}
		\caption{Simplified $\bar{G}$}
		\label{fig:finalNetworkGraph1}
	\end{subfigure}
	\caption{Constructing $\bar{G}$ by utilizing a set of virtual vertices}
	\label{fig:G_construction}
\end{figure}

In the following, we present a detailed example of constructing ${\bar{G}}$ for the network in Fig. \ref{fig:networkGraph1}. In order to find a path from $ v_1 $ to $ v_5 $ with an $\text{SFC} = (v_1, \phi_1, \phi_2, v_5)$, we start by transforming the graph ${G}$ into a new graph ${\bar{G}}$ as in Fig. (\ref{fig:newNetworkGraph1}). For each vertex $v_k$ in Fig. \ref{fig:networkGraph1}, we create three  virtual vertices $v_k^0, v_k^1,$ and $v_k^2$ because the number of network functions specified by the SFC is two. The edges in ${\bar{G}}$ are constructed as follows. For $v_1^0$, we connect it to $v_2^2 \in \bar{\V}_{v_2}$ as all the functions are available in $v_2$. However, we connect $v_1^0$ to $v_4^1$ as only function $\phi_1$ is available at node $v_4$. The edge cost of $(v_1^0, v_2^2)$ and $(v_1^0, v_4^1)$ is set to that of edge $(v_1, v_2)$ and $(v_1, v_4)$ in ${G}$, respectively. We repeat the same procedure for each vertex and obtain initial ${\bar{G}}$. Then, we obtain pruned ${\bar{G}}$ by repeatedly removing the vertices that do not have any incoming edges (except for the source) or that do not have any outgoing edges (except for the destination). The final graph is shown in Fig. (\ref{fig:finalNetworkGraph1}). The shortest path from $v_1$ to $v_5$ is $(v_1^0,   v_3^0,   v_4^1,   v_3^2,   v_5^2)$ (see Fig. (\ref{fig:finalNetworkGraph1})). This path is mapped to $(v_1,   v_3,   v_4,   v_3,   v_5)$ in the original graph ${G}$.

\subsection{Algorithm Analysis}
We start by proving the correctness of the proposed algorithm in Theorem \ref{Theorem:SP}. Then, we show the performance of the pruning step in Lemma \ref{Lemma:pruning}.
\begin{theorem}
	A shortest path to the destination vertex $v_d^r$ in $\bar{G}$ is a shortest path to the destination vertex $v_d$ in $G$ that satisfies the given SFC constraint.
	\label{Theorem:SP}
\end{theorem}
We first establish the following Lemma , which will be used in the proof of Theorem \ref{Theorem:SP}.
	\begin{lemma}
		\label{lemma:SFC}
		A path to any vertex $v_k^i$ in $\bar{G}$ is also a path to $v_k$ in $G$ that guarantees satisfying the partial service function chain $(\phi_1, \dots, \phi_{i})$.
	\end{lemma}
	\begin{proof}
		We will prove Lemma \ref{lemma:SFC} by induction. \newline
		Let denote a path to vertex $v_k^i$ as $(v_s^0, \dots, v_l^j, v_k^i)$. The base case is the trivial case, which is that a path to $v_s^0$ satisfies zero netwok functions. The induction hypothesis is that a path to $v_l^j$ is a path to $v_l$ in $G$  that satisfies the partial service function chain $(\phi_1, \dots, \phi_j)$.
		We want to show that a path to $v_k^i$ is also a path to $v_k$  in $G$ that satisfies the partial service function chain $(\phi_1, \dots, \phi_i)$. Based on our construction of $\bar{G}$, the edge $(v_l^j, v_k^i)$ can be established if the set $\{\phi_{j+1}, \dots, \phi_i\}$ is a subset of $\Phi_{v_k}$. Moreover, by the induction hypothesis, a path to $v_j$ satisfies network functions in $(\phi_1, \dots, \phi_j)$. As a result, a path to $v_k$ includes the network functions supported by $v_l$, which is $(\phi_1, \dots, \phi_j)$, and the network functions supported by $v_k$, which is $\{\phi_{j+1}, \dots, \phi_i\}$. That would be the chain $(\phi_1, \dots, \phi_i)$. This completes the induction step. \newline		
	\end{proof}
\begin{proof}[Proof of Theorem 1]
	Let $\mathcal{\bar{P}}$ denote the set of possible paths from $v_s^t$ to $v_d^r$ in $\bar{G}$.
	Based on Lemma \ref{lemma:SFC}, all the paths in $\mathcal{\bar{P}}$ satisfies the $r$ network functions. The set of paths $\mathcal{\bar{P}}$ can be mapped to  a set of paths $\mathcal{P}$ in $G$. We can use any shortest path algorithm to select the shortest path ${\bar{p}} \in \mathcal{\bar{P}}$, and its mapping in $\mathcal{P}$ will be the shortest path in $G$.
\end{proof}

Next, in Lemma \ref{Lemma:pruning}, we show that the pruning step will reduce the size of the new constructed graph $\bar{G}$, which leads to an efficient computation of an SFC-constrained shortest path.
\begin{lemma}
	If the probability of availability of each network function at any node is $z$, then the pruning step will remove at least $0.5z|\bar{V}|$ nodes and $0.5z|\bar{E}|$ edges of the initial $\bar{G}$.
	\label{Lemma:pruning}
\end{lemma}
\begin{proof}
	For any pair of nodes $(v_l, v_k)$, if the probability that function $\phi_i$ is available at node $v_k$ is $z$, then  virtual vertex $v_l^{i-1}$ will be connected to $v_k^i$ with probability $z$. Similarly, for any other pairs of the form $(v_m, v_k)$. In such cases, virtual vertex $v_k^{i-1}$ will have no incoming edges and will be removed in the pruning step. So, for any pair of vertices, we can remove one vertex with probability $z$, i.e., reducing the number of vertices by a half. So, the overall number of removed vertices is at least $0.5z|\bar{V}|$. If we assume that the edges are uniformly distributed in the graph, then for a certain percentage of the removed nodes, we will remove their corresponding edges, which results in a similar reduction in the number of edges. A special case is when $z=1$, then the number of vertices and edges will be the same of the original graph, $G$.
\end{proof}

\subsection{Algorithm Complexity}
For the initial $\bar{G}$, the number of vertices will be at most $|\bar{\V}| = (r+1)|\V|$, and the edges $|\bar{E}| = (r+1)|\E|$. The complexity of constructing $\bar{G}$ is $O(|\bar{\V}|^2)$. Then, after the pruning step, the number of vertices (resp., edges) will be at most $|\V'| = 0.5z|\bar{\V}|$ (resp., $|\E'| = 0.5z|\bar{\E}|$), where $z$ is the probability of availability of network functions at any vertex. We construct $\bar{G}$ only once for a given SFC constraint. Then, if we use Dijkstra's algorithm to find an SFC-constrained shortest path on the pruned $\bar{G}$, the complexity is $O(|\E'| \log |\V'|)$. 

\subsection{Throughput-Optimal Routing with SFC Constraints}
A general framework for throughput-optimal routing, called Universal Max-Weight (UMW) policy, has been proposed in \cite{OptimalControl}. In the UMW policy, each source maintains a virtual queue for each physical queue in the network. When a packet arrives, it computes the shortest path based on the length of the virtual queues. UMW policy considers different types of traffic. In the case of unicast traffic, we are able to extend the UMW policy for traffic with SFC constraints. We do that by integrating our SFC-constrained shortest path algorithm with the UMW policy.  The UMW policy with our shortest path algorithm remains throughput optimal. The proof of throughput optimality is omitted here since it is exactly the same as that of \cite{OptimalControl}. We provide some numerical results and interested readers can refer to \cite{OptimalControl} for further details about the UMW policy.

\section{SFC-constrained Maximum Flow (SFC-MF)}
\label{sec:SFC-MF}
The maximum flow problem is a classic problem, where the maximum possible flow from a source node to a destination node needs to be computed. Classic maximum flow algorithms ensure that the flow on each edge does not exceed its capacity, and the flow conservation constraint is satisfied. Service function chains constraints, which require each flow to traverse a
set of network functions in a pre-specified order before reaching its destination, make this problem more challenging. In this paper, this new SFC-constrained Maximum Flow problem is referred to as the SFC-MF problem. Note that the classic maximum flow algorithms (e.g., Ford-Fulkerson) are not directly applicable to SFC-MF due to the new constraints. We consider an undirected graph in which edges can be used to send flow in either direction, but the total flow in both directions cannot exceed the edge capacity. Also, we consider that each function has only one instance in the network. For an SFC constraint defined as $\text{SFC} = (s, \phi_1, \dots, \phi_r, d)$, we define a commodity $\alpha_i$ for each segment of the SFC, i.e., commodity $\alpha_1$ has a source $s$ and destination $\phi_1$, while commodity $\alpha_{r+1}$ has a source $\phi_r$ and destination $d$. We use $\mathcal{P}_{\alpha_i}$ to denote the set of all possible paths for commodity $\alpha_i$. The capacity of edge $e$ is denoted by $c_e$. Also, we let $x_p$ denote the amount of flow sent over path $p$. We formulate the SFC-MF problem as follows.
\begin{align}
	& \max \quad \lambda, \label{eq:fractional_MC} \\
	 \text{subject to} \nonumber\\
	& \sum_{p \in \mathcal{P}_{\alpha_i} } x_p \geq \lambda, \quad \forall \alpha_i, \label{eq:total_flow} \\
	& \sum_{\alpha_i} \sum_{p:e \in p, p \in \mathcal{P}_{\alpha_i}} x_p \leq c_e, \quad \forall e \in \E, \label{eq:edge_capacity_MF}\\
	& x_p\geq 0, \quad \forall p\in \mathcal{P}_{\alpha_i},
\end{align}
 where Eq. \eqref{eq:total_flow} is the total flow for commodity $\alpha_i$ over all of its possible paths. Eq. \eqref{eq:edge_capacity_MF} ensures that the amount of flow for all commodities over an edge does not exceed its capacity. 
 
 Problem \eqref{eq:fractional_MC} is a special case of the fractional multicommodity flow (FMCF) problem when the demand of all commodities is one. It can be solved using any LP algorithm, or approximation algorithm \cite{Karakostas2008, Fleischer1999}. An interesting use case is when a service provider needs to virtualize a particular network function in its network during an early stage of NFV deployment. So, we consider the problem of computing the maximum flow with the constraint that all packets must pass through this new VNF. For instance, suppose that we have an SFC constraint as in Fig. \ref{fig:must_stop_main}(\subref{fig:SFC_constraint}), and we want to virtualize only function $\phi_l$. For the network shown in Fig. \ref{fig:must_stop_main}(\subref{fig:must_stop}), most of the nodes support some PNFs. So, we can pick a node, say node $v_5$, to become a virtualized node and host the VNF $\phi_l$. Then, we may need to compute the maximum flow from any PNF to another PNF through node $v_5$. For this case, we propose an elegant combinatorial algorithm based on the Ford-Fulkerson algorithm in the following subsection.

 \subsection{Maximum Flow with One Must-stop Node Algorithm}

 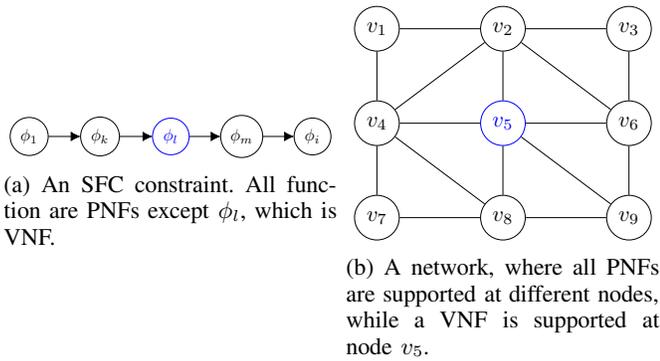
\begin{figure}
 	\begin{subfigure}{0.50\linewidth}
 		\centering	
 		\resizebox{\linewidth}{!}{
 				\begin{tikzpicture}[transform shape]
 				\node[vertex](p1) at (-3, 2) {$ \phi_1 $};
 				\node[vertex](pk) at (-1.5, 2) {$ \phi_k $};
 				\node[vertex, color=blue](v) at (0, 2) {$ \phi_{l} $};
 				\node[vertex](pm) at (1.5, 2) {$ \phi_{m}$};
 				\node[vertex](pi) at (3, 2) {$ \phi_i $};
 				\begin{scope}[every path/.style={->}, every node/.style={inner sep=1pt}]
 				\path (p1) edge [ anchor= south] node {} (pk);
 				\path (pk) edge [ anchor= south] node {} (v);
 				\path (v) edge [ anchor= south] node {} (pm);
 				\path (pm) edge [ anchor= south] node {} (pi);
 				
 				\end{scope} 
 				\end{tikzpicture}
 		}
 		\caption{An SFC constraint. All function are PNFs except $\phi_l$, which is VNF.}
 		\label{fig:SFC_constraint}
 	\end{subfigure}
 	\begin{subfigure}{0.47\linewidth}
 		\centering
 		\resizebox{\linewidth}{!} {
 			\begin{tikzpicture}[transform shape]
 			\node[vertex](v_1) at (0, 1) {$v_1$};
 			\node[vertex](v_2) at (2, 1) {$v_2$};
 			\node[vertex](v_3) at (4, 1) {$v_3$};
 			\node[vertex](v_4) at (0, -0.5) {$v_4$};
 			\node[vertex, color=blue](v_5) at (2, -0.5) {$v_5$};
 			
 			\node[vertex](v_6) at (4, -0.5) {$v_6$};
 			\node[vertex](v_7) at (0, -2) {$v_7$};
 			\node[vertex](v_8) at (2, -2) {$v_8$};
 			\node[vertex](v_9) at (4, -2) {$v_9$};

 			\begin{scope}[every path/.style={-}, every node/.style={inner sep=1pt}]
 			\path (v_1) edge [ anchor= south] node {} (v_2);
 			\path (v_1) edge [ anchor= south] node {} (v_4);
 			
 			\path (v_2) edge [ anchor= south] node {} (v_3);
 			\path (v_2) edge [ anchor= south] node {} (v_5);
 			\path (v_2) edge [ anchor= south] node {} (v_6);
 			\path (v_3) edge [ anchor= south] node {} (v_6);
 			
 			\path (v_4) edge [ anchor= south] node {} (v_7);
 			\path (v_4) edge [ anchor= south] node {} (v_5);
 			\path (v_4) edge [ anchor= south] node {} (v_8);
 			\path (v_5) edge [ anchor= south] node {} (v_6);
 			\path (v_5) edge [ anchor= south] node {} (v_8);
 			\path (v_5) edge [ anchor= south] node {} (v_9);
 			\path (v_6) edge [ anchor= south] node {} (v_9);
 			\path (v_2) edge [ anchor= south] node {} (v_4);
 			\path (v_8) edge [ anchor= south] node {} (v_9);
 			\path (v_7) edge [ anchor= south] node {} (v_8);
 			
 			\end{scope} 
 			\end{tikzpicture}
 		}
 		\caption{A network, where all PNFs are supported at different nodes, while a VNF is supported at node $v_5$.}
 		\label{fig:must_stop}
 	\end{subfigure}
 	\caption{Example of an SFC constraint where only one function, $\phi_l$, will be virtualized at node $v_5$.}
 	\label{fig:must_stop_main}
 \end{figure}

 \begin{figure}
 	\centering
 	\begin{tikzpicture}[transform shape]

 	\node[vertex](S) at (0, 0) {$ s $};	
 	\node[vertex](A) at (1.5, 0) {$ a $};
 	\node[vertex](B) at (3, 0){$ b $};	
 	\node[vertex](D) at (4.5, 0) {$ d $};
 	\node[vertex](T1) at (2.25, 1) {$ t $};	
 	\node[vertex](T2) at (2.25, -1.5) {$ T $};	
 	\begin{scope}[every path/.style={-}, every node/.style={inner sep=1pt}]
 	\path (S) edge [ anchor= south] node {2} (A);
 	\path (A) edge [ anchor= south] node {2} (B);
 	\path (A) edge [ anchor= east] node {1} (T1);
 	\path (B) edge [ anchor= west] node {2} (T1);
 	\path (B) edge [ anchor= south] node {2} (D);
 	\path (S) edge [ anchor= south, densely  dotted] node {$\infty$} (T2);
 	\path (D) edge [ anchor= south, densely  dotted] node {$\infty$} (T2);
 	\end{scope} 
 	\end{tikzpicture}
 	\caption{A network with must-stop node $t$ and an added virtual node $T$.}
 	\label{fig:MF-T}
 \end{figure}
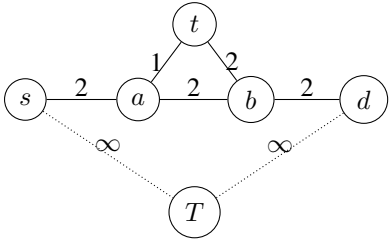
  We denote the maximum flow with one must-stop node $t$ as $F_{s,d}^t$. Let $\mathcal{P}_{st}$  (resp.,  $\mathcal{P}_{td}$) be the set of all paths between node $s$ and $t$ (resp., $t$ and $d$). Then, $F_{s, d}^t$ can be defined as the maximum feasible flow that can be sent simultaneously from $s$ to $t$, over $\mathcal{P}_{st}$, and from $t$ to $d$, over $\mathcal{P}_{td}$. We use $F_{s,t}$ (resp., $F_{t, d}$) to denote the maximum feasible flow that can be sent over $\mathcal{P}_{st}$ (resp., $\mathcal{P}_{td}$). To compute $F_{s,d}^t$, we start by adding a virtual node $T$ to the graph and connect it to the source, $s$, and destination, $d$, with infinite capacity. Then, we compute the standard maximum flow for the following cases: (1) The maximum flow from $t$ to $T$, divided by 2, denoted by $F_{t, T}/2$; (2) $F_{s, t}$; (3) $F_{t, d}$. Then, $F_{s,d}^t$ is the minimum among the above three quantities. We can define $F_{t, T}$ as the maximum feasible flow that can be sent from $t$ to $T$ over both $\mathcal{P}_{st}$ and $\mathcal{P}_{td}$. Since we consider undirected graphs, then, the maximum flow from $s$ to $t$ and from $t$ to $s$ will be the same. In the following, we show an example of how to compute the maximum flow through node $t$ for the network shown in Fig. \ref{fig:MF-T}. First, we connect a virtual node $T$ to nodes $s$ and $d$ with infinite capacity. Then, we compute these three quantities, $F_{t, T}/2 , F_{s,t},$ and $F_{t, d}$, which will be 1.5, 2, and 2, respectively. It can be verified that the minimum of them is 1.5, which is equal to the maximum flow from node $s$ to $d$ through node $t$. 
  
  We prove this result in Lemma \ref{Lemma:one_stop}.     
\begin{lemma}
	$F_{s,d}^t = \min\{ F_{t,T}/2 , F_{s,t}, F_{t, d}\}$.
	\label{Lemma:one_stop}
\end{lemma}

\begin{proof}
	We start by showing that $F_{t,T}/2, F_{s,t}, F_{t, d}$ are upper bounds for $F_{s,d}^t$. Then, we show that the minimum of these three values is also a lower bound of $F_{s,d}^t$. We state the following upper bounds. 
	\begin{itemize}
		\item $F_{s,d}^t \leq F_{t, T}/2$ as we cannot send more than half of $F_{t, T}$ simultaneously over $\mathcal{P}_{st}$ and $\mathcal{P}_{td}$.
		\item $F_{s,d}^t \leq F_{s,t}$ because all flow from node $s$ should reach node $t$ before reaching node $d$.
		\item $F_{s,d}^t \leq F_{t, d}$ because all flow reaching node $d$ should pass by node $t$ first. 
	\end{itemize}
	Then, we prove that the minimum of these upper bounds is also a lower bound. We present three cases: each case corresponds to when one of the upper bounds is the minimum as follows.
	
	Case I: when {$F_{t, T}/2$ is the minimum}, we will show that we can always send $F_{t, T}/2$ simultaneously over $\mathcal{P}_{st}$ and $\mathcal{P}_{td}$. First, let $F_{t, T} = N_1 + N_2$, where $N_1$ (resp., $N_2$) is the amount of flow that is sent over $\mathcal{P}_{st}$ (resp., $\mathcal{P}_{td}$) paths in the current realization. We have three subcases: A) $N_1 = N_2$, B) $N_1 > N_2$, and C) $N_1 < N_2$, which are discussed in the following.
	
	Case I-A: when $N_1 = N_2$. This is a trivial case. It is easy to see that we can send $F_{t, T}/2$ simultaneously over $\mathcal{P}_{st}$ and $\mathcal{P}_{td}$.  
	
	Case I-B: when $N_1 > N_2$, i.e., $N_1 = F_{t, T}/2 + c$ and $N_2 = F_{t, T}/2 - c$, for a positive $c$. We will show that we can always remove $c$ units of flow from $\mathcal{P}_{st}$ paths and send the same amount over $\mathcal{P}_{td}$ paths, where this reallocation of flows will make $N_1 = N_2$. We use $x_p$ to denote the current flow over path $p$. Also, for one realization of the maximum flow over paths in $\mathcal{P}_{td}$, we use $x'_p$ to denote the  amount of flow sent over every path $p$ in $\mathcal{P}_{td}$. Due to the intersection of edges of paths in $\mathcal{P}_{st}$ and $\mathcal{P}_{td}$, the capacity of these edges is shared by such paths. So, a flow over a path in $\mathcal{P}_{st}$ may affect the amount of flow over some paths in $\mathcal{P}_{td}$, i.e., making $x_p$ less than $x'_p$ for such paths. Since $\sum_{p \in \mathcal{P}_{td}} x'_p = F_{t,d}$, which is greater than $F_{t, T}/2$, then, it is feasible to reallocate $c$ units of flow to over $\mathcal{P}_{td}$ paths; the details are provided in the following.
	
	Since $F_{t,d}$ is greater that $N_2$, then, we can find  a set of paths in $ \mathcal{P}_{td}$ that satisfy $x_{p} < x'_{p}$, so, we pick one of them and denote it as $p_i$. Then, for the set of paths in $\mathcal{P}_{st}$ that intersect with some edges of path $p_i$, we select a path $p_j$ that has positive $x_{p_j}$ and intersects with path $p_i$ at an edge that is the nearest to node $d$. Let $p_i = (e_1, \dots , e_f, e_{f+1}, \dots, e_n)$, and $p_j = (\bar{e}_1, \dots ,  \bar{e}_l, \bar{e}_{l+1}, \dots, \bar{e}_m)$, with $e_f \in p_i$ is the closest edge to node $d$ that intersects with $\bar{e}_l \in p_j$, i.e., $e_f = \bar{e}_l$. Since $e_f$ is the closest edge to node $d$ that intersects with a path in $\mathcal{P}_{st}$, then, edges $(e_{f+1},  \dots, e_n)$ can support $x'_{p_i} - x_{p_i}$ units of flow.  Next, we define $h = \min(x'_{p_i} - x_{p_i}, x_{p_j}, c)$ and cancel this amount of flow over path $p_j$. As a result, edges $(\bar{e}_{l}, \bar{e}_{l+1}, \dots ,\bar{e}_m)$, which are part of $p_j$, will be able to support an additional $h$ units of flow.  Finally, we construct a path $p_c$ by taking edges $(\bar{e}_{l}, \bar{e}_{l+1}, \dots ,\bar{e}_m)$ from $p_j$ in reverse order, i.e., $(\bar{e}_m, \dots, \bar{e}_{l+1}, \bar{e}_{l})$,  and edges $(e_{f+1},  \dots, e_n)$ from $p_i$, and forming a new path $p_c = (\bar{e}_m, \dots, \bar{e}_{l+1}, \bar{e}_{l}, e_{f+1}, \dots, e_n)$. We can see that over path $p_c$ we can send additional $h$ units of flow from node $t$ to $d$. 
	
	We subtract $h$ from $c$ and repeat the same process until the value of $c$ becomes zero. This is feasible because as long as the amount of flow sent over $\mathcal{P}_{td}$ is less than $F_{t, T}/2$, then, we can find a path $p$ in $\mathcal{P}_{td}$ with $x_p < x'_p$.  Also, we assume integral capacity of edges, so, by using Ford-Fulkerson algorithm, the value of $N_1, N_2, x_p, x'_p$, and $F_{t, T}$ are integral. In addition, the value of $c$ is a multiple of $0.5$ because $c = N_1 - F_{t, T}/2$. From this, we conclude that the value of $h$ is a multiple of $0.5$. So, in at most $2 \times F_{s,d}^t$ iterations, we can make $N_1 = N_2$ by the above procedure.  
	
	Case I-C: this is a symmetric case of Case I-B.
	
	Case II: when {$F_{s, t}$ is the minimum}, we will show that $F_{s, t}$ can be sent simultaneously over $\mathcal{P}_{st}$ and $\mathcal{P}_{td}$ paths. Since $F_{s, t}$ is the maximum flow that can be sent over paths in $\mathcal{P}_{st}$ and we have that $F_{t, T}/2 \geq F_{s,t}$, then $F_{t, T} \geq 2F_{s,t}$. That means the amount of flow sent over $\mathcal{P}_{td}$ should be at least $F_{s,t}$.
	
	Case III: when {$F_{t, d}$ is the minimum}. The proof of this case follows the same argument as that of Case II.
\end{proof}
	 
After we find the value of the maximum flow, it remains to find the actual flow on each edge and the direction of the flow. To do that, we again use a virtual node $T$ and connect it to the source $s$ and destination $d$, but with a capacity of $F_{s,d}^t$. Then, we compute the maximum flow in this new graph from $t$ to $T$. When a path includes edge $(s,  T)$, then the flow on the links along this path is reversed. That would give us the amount of flow and direction on each edge.

\section{Virtual Network Functions Placement}
\label{sec:VSF_placement}
In this section, we are interested in the question of how to place VNFs such that the value of the maximum flow with SFC constraints is equal to the value of the original maximum flow without SFC constraints. The maximum flow under SFC constraint is not expected to remain as the original maximum flow, depending on the placement of the required network functions specified by the SFC constraint. Moreover, in order to minimize the total operational expenses of adding commodity servers in the network to support VNFs, we aim to minimize the number of network nodes where these functions will be placed. We assume that all VNFs can be hosted at any node. We start by formulating the problem and proving its NP-hardness.

We let $\mathcal{P'}_{sd}$ denote the set of admissible paths from node $s$ to $d$ for a given SFC constraint. We want to select a minimum number of nodes such that the total flow over the admissible paths $\mathcal{P'}_{sd}$ is the original maximum flow (without any SFC constraint). Define $k_{i}$ as a binary variable to denote whether node $i$ hosts the required VNFs (i.e., node $i$ is a virtualized node). Also, $x_p$ denotes the amount of flow over path $p$. We use $F_{s,d}$ to denote the original maximum flow. The problem can be formulated as follows.
\begin{align}
&\min \sum_{i \in \V \backslash \{s, d\}}  k_{i}, \label{eq:placement_main}\\
\text{subject to} \nonumber \\
&\sum_{p \in \mathcal{P'}_{sd}} x_p = F_{s,d}, \label{eq:paths_flow} \\
&\sum_{p: e\in p} x_p \leq c_e, \quad \forall e \in \E, \label{eq:edge_capacity1}\\
& x_p \geq 0, \quad \forall p \in \mathcal{P'}_{sd}, \\
&k_{i} \in \{0, 1\}, \quad \forall i,
\end{align} 
where the objective is to minimize the number of virtualized nodes, described by \eqref{eq:placement_main}. Eq. \eqref{eq:paths_flow} ensures that flow over all admissible paths equals the original maximum flow, while Eq. \eqref{eq:edge_capacity1} ensures that the total flow over an edge does not exceed its capacity.

Next, we will show  in Lemma \ref{Lemma:placement} that this problem is NP-hard based on a reduction from the classic set-cover problem.
\begin{lemma}
	The minimum placement of VNFs to achieve the original maximum flow is NP-hard.
	\label{Lemma:placement}
\end{lemma}
\begin{proof}

 We prove that by a reduction from the minimum set-cover problem. In the set-cover problem, we are given a set $\mathcal{M}$ of $n$ elements, $\mathcal{M} =\{m_1, m_2, \dots, m_n\}$, and a collection of subsets $\mathcal{S}=\{s_1, s_2, \dots, s_l\}$, where each $s_i$ is a subset of $\mathcal{M}$ and the union of all subsets in $\mathcal{S}$ is $\mathcal{M}$. The minimum set-cover problem is to find the minimum number of subsets in $S$ such that their union is $\mathcal{M}$. So, given an instance of the set-cover problem $(\mathcal{S}, \mathcal{M})$, we will reduce it to our problem. We construct a graph $G =(\V, \E)$. Each vertex in $\V$ corresponds to a subset in $\mathcal{S}$, plus one vertex as a source and one vertex as a destination. For each element $m_i \in \mathcal{M}$, we construct a path that connects the subsets containing element $m_i$ in any order and connect the first node in this path to the source and the last node to the destination. The capacity of each edge along the constructed path is set to one; if an edge has been established before, its capacity is increased by one. So, each constructed path will contribute a unit flow to the maximum flow. If we can solve the minimum set-cover problem, then the corresponding vertices in $G$ will cover all possible flow because each element in the set  corresponds to a unit of flow. Similarly, if we can find the minimum number of vertices to achieve the maximum flow, then, each unit of flow will pass by one of these vertices, so the corresponding subsets will  cover all elements in $\mathcal{M}$. 	
\end{proof}

The hardness of the problem comes from two parts. First, the maximum flow can be achieved through a different set of paths, which are hard to list, and each set of paths may yield a different placement. Moreover, finding the minimum number of nodes to cover a set of paths is also NP-hard \cite{Sang}. Also, it is worth noting that if we consider service function chains where functions cannot be hosted at one node, then the problem becomes harder. 

The formulation in \eqref{eq:placement_main} has an exponential number of constraints as it requires to list all admissible paths in Eq. \eqref{eq:paths_flow}, which could be exponential. So, we provide an equivalent formulation for the problem  that can be solved using any Integer Linear Programming (ILP) solver. 

\subsection{ILP Formulation of VNFs Placement}
First, we introduce the following notations. We use $f_{ij}^0$ (resp., $f_{ij}^1$) to denote the amount of unprocessed (resp., processed) flow on link $(i, j)$ by the VNFs of the SFC constraint. Also, we use $\delta_i^+$ (resp.,  $\delta_i^-$)  to denote the set of incoming (resp., outgoing) edges of node $i$. Finally, we use $c_{i, j}$ to denote the capacity of edge $(i, j)$. The problem becomes:
\begin{align}
& \min \sum_{i \in \V \backslash \{s, d\}}  k_{i}, \label{eq:placement_equivlent} \\
\text{subject to} \nonumber \\
&\sum_{j \in \delta_i^+} f_{ji}^0 + \sum_{j \in \delta_i^+} f_{ji}^1 = \sum_{k \in \delta_i^-} f_{ik}^0 + \sum_{k \in \delta_i^-} f_{ik}^1, 	\label{eq:conservation} \\
& \sum_{k \in \delta_i^-} f_{ik}^1 = k_{i} \sum_{j \in \delta_i^+} f_{ji}^0 + \sum_{j \in \delta_i^+} f_{ji}^1, 	\label{eq:new_conservation} \\
&\sum_{k \in \delta_s^-} f_{sk}^0 = F_{s,d}, 	\label{eq:max_source} \\
& \sum_{j \in \delta_d^+} f_{jd}^1 = F_{s,d}, \label{eq:max_dst} \\
& f_{ij}^0 + f_{i, j}^1 \leq c_{i, j}, \quad \forall (i, j) \in \E, \label{eq:edge_capacity} \\
& k_{i} \in \{0, 1\}, \label{eq:binary} \\
& f_{ij}^t \geq 0, \quad \forall i, j \in \V, t \in \{0, 1\},
\end{align}
where, when not specified, $ i \in \V \backslash \{s, d\}$.
Eq. \eqref{eq:conservation} is the standard flow conservation constraint, while Eq. \eqref{eq:new_conservation} is to ensure that the amount of processed flow leaving node $i$ is either an unprocessed flow that is processed by node $i$ or a flow that has been processed by other nodes before entering node $i$. Eq. \eqref{eq:max_source} and \eqref{eq:max_dst} ensure that the unprocessed flow from the source will reach the destination as a processed flow, and the amount of this flow equals the original maximum flow. Eq. \eqref{eq:edge_capacity} ensures that the amount of both processed and unprocessed flow over each edge does not exceed its capacity. A node is either an intermediate node or virtualized node (i.e., hosts all the VNFs), which is considered by the binary variable $k_i$ in Eq. \eqref{eq:binary}. 

The formulation in \eqref{eq:placement_equivlent} can be solved by any ILP solver, which we were able to solve for large instances (e.g., for network with 100 nodes) in a few minutes. Moreover, we show by simulations that the maximum flow can be achieved by placing the VNFs at a small number of nodes even when the graph is large. This indicates that the operators may be able to introduce VNFs in their networks at a low starting cost without impacting the amount of flow that can be sent.

\section{Numerical results}
\label{sec:numerical_analysis}
\subsection{SFC-Constrained Shortest Path Results}
\begin{figure}
	\centering
	\begin{subfigure}{.5\linewidth}
		\centering
		\includegraphics[width=1\linewidth]{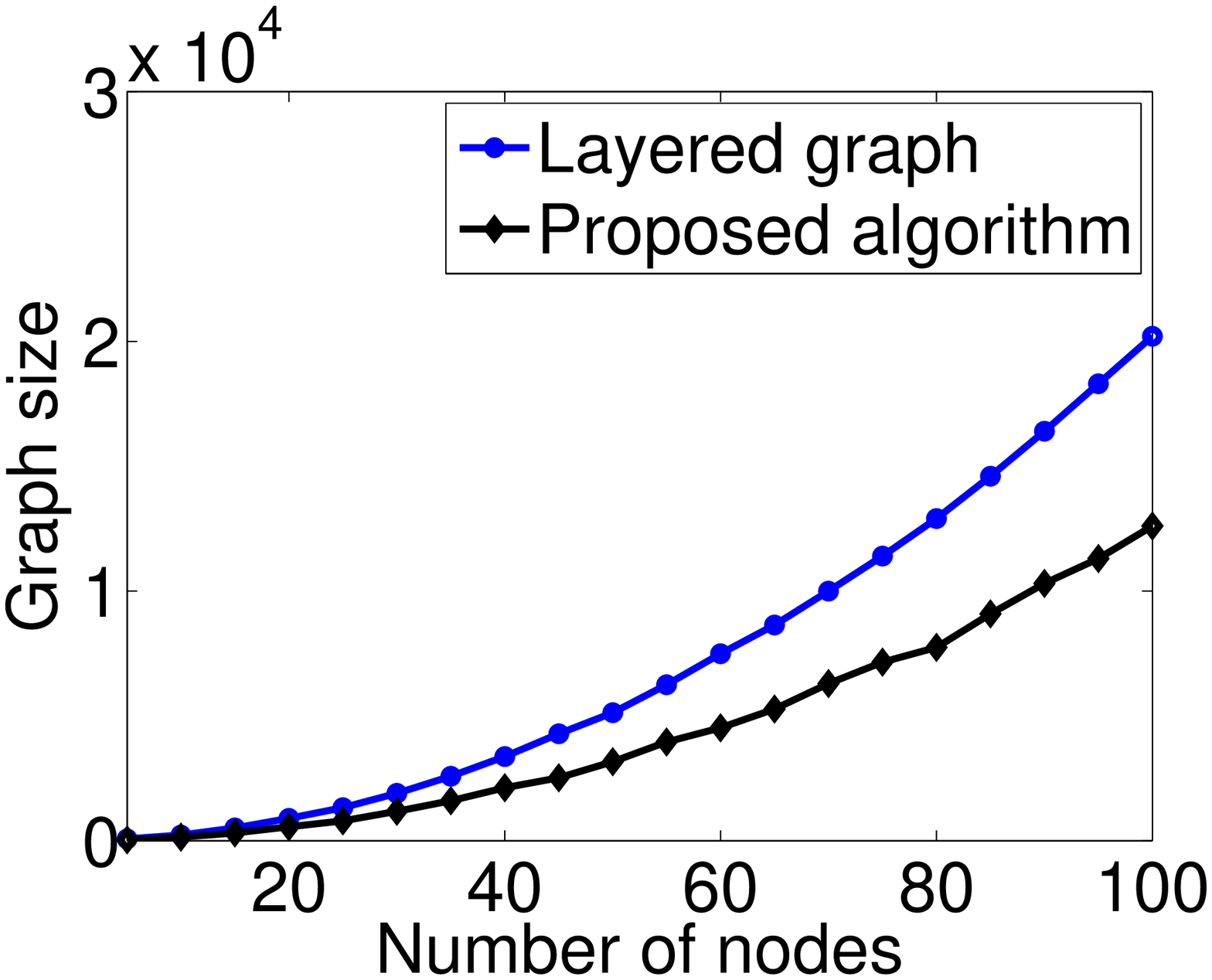}
		\caption{Graph size with different number of nodes.}
		\label{fig:our_layered1}
	\end{subfigure}%
	\begin{subfigure}{.5\linewidth}
		\centering
		\includegraphics[width=1\linewidth]{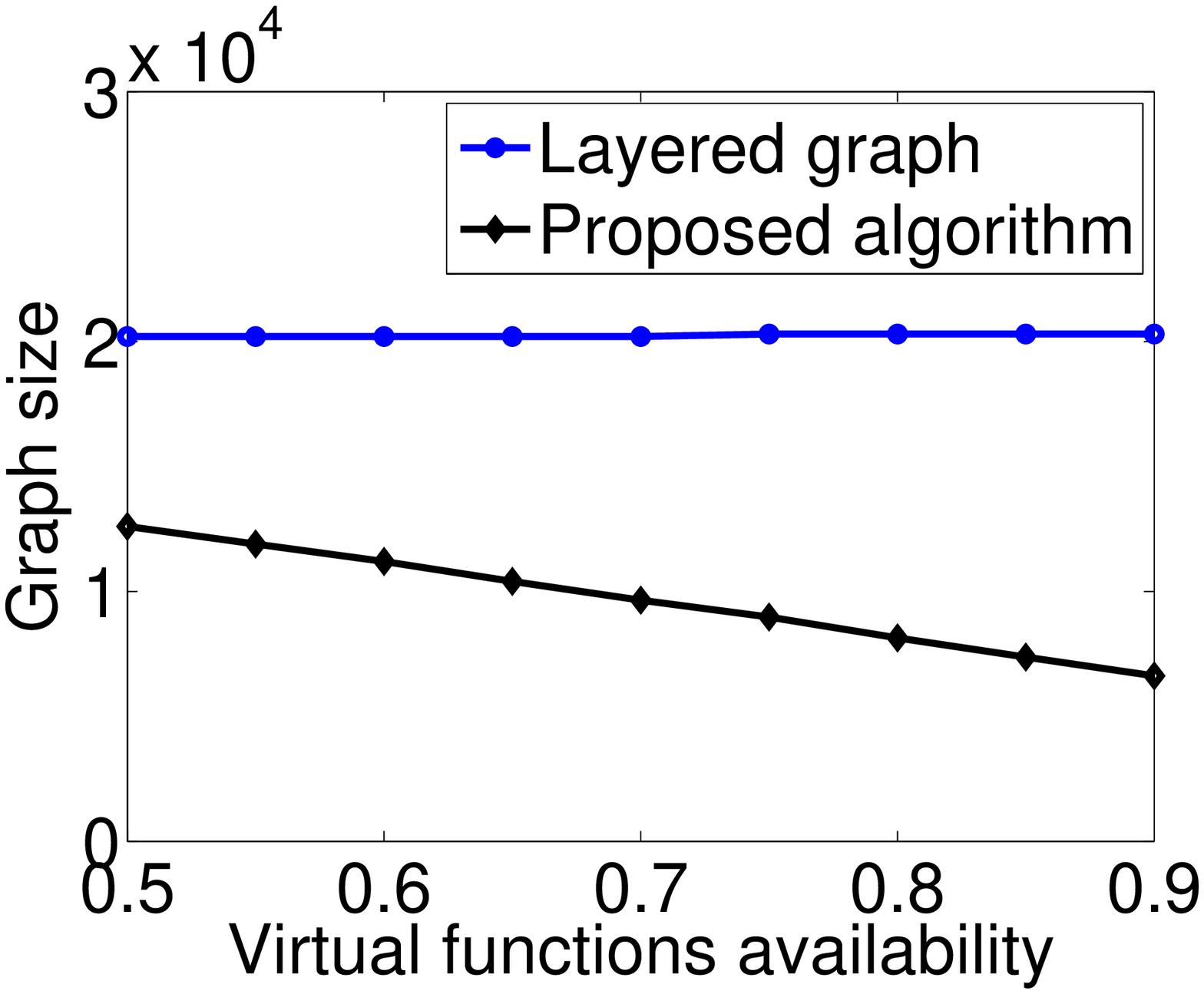}
		\caption{Graph size with different functions availability probability.}
		\label{fig:our_layered2}
	\end{subfigure}
	\caption{Graph size of our approach compared to the layered approach.}
	\label{fig:our_layered}
\end{figure}

In this subsection, we evaluate the proposed SFC-constrained shortest path algorithm, and compare its performance with the layered graph that was proposed in \cite{Dwaraki2016, Cao2014}. The layered graph is constructed by replicating the original graph $r+1$ times, where each replication is a layer. Each layer $i \leq r$ is connected to layer $i+1$ by connecting the nodes that host the $i$-th network function in layer $i$ to the same set of nodes in layer $i+1$. Then, a shortest path is computed from the source in layer one to the destination in layer $r+1$. We consider an SFC represented as $(s, \phi_1, \phi_2, \phi_3, d)$, and each function, $\phi_i$, is available at each node with a probability of $0.5$. We compare the results in term of the constructed graph size, which is the number of nodes plus the number of edges. We should note that the smaller the size of the constructed graph, the lower the complexity of any shortest path algorithm. We repeat each experiment for 10 times, and report the average result. In Fig. \ref{fig:our_layered}(\subref{fig:our_layered1}), we can see that the size of the layered graph is larger than the graph constructed by our approach. Moreover, as the probability of availability of functions in each node increases, our approach has a very small graph size compared to the layered graph, as shown in Fig. \ref{fig:our_layered}(\subref{fig:our_layered2}).

\begin{figure}
	\begin{subfigure}{0.56\linewidth}
		\centering	
		\resizebox{\linewidth}{!}{
			\begin{tikzpicture}[transform shape]
			\node[vertex, label=above:\textcolor{red}{$s_1$}](v_1) at (0, 0) {$ v_1 $};
			\node[vertex, label={[xshift=-1.5cm, yshift=0.6cm]0:{{$\{\phi_1\}$}}}] (v_2) at (2, 1) {$ v_2 $};
			\node[vertex, label=below:{$\{\phi_1, \phi_3\}$}](v_3) at (2, -1){$ v_3 $};
			\node[vertex, label={[xshift=-0.5cm, yshift=0.6cm]0:{{$\{\phi_2, \phi_3\}$}}}] (v_4) at (5, 1) {$ v_4 $};
			\node[vertex, label=below:{$\{\phi_2\}$}, label=right:\textcolor{red}{$d_2$}](v_5) at (5, -1) {$ v_5 $};
			\node[vertex, label=above:\textcolor{red}{$d_1$}](v_6) at (7, 0) {$ v_6 $};
			\node[vertex , label=left:\textcolor{red}{$s_2$}](v_7) at (3, 2.7) {$  v_7 $};
			
			\begin{scope}[every path/.style={->}, every node/.style={inner sep=1pt}]
			
			\draw (v_1) edge [  anchor= south] node {} (v_2);
			
			\draw (v_1) -- node [anchor=south] {} (v_3);
			
			\draw (v_2) -- node [anchor= south] {} (v_4);
			\draw (v_3) -- node [anchor= south] {} (v_5);
			
			\draw (v_4) -- node [anchor= south] {} (v_6);
			\draw (v_5) -- node [anchor= south] {} (v_6);
			\draw (v_7) -- node [anchor= south] {} (v_2);
			\draw (v_7) -- node [anchor= south] {} (v_4);
			\draw (v_3) -- node [anchor= south] {} (v_4);
			
			\draw (v_2) -- node [anchor= south] {} (v_3);
			\draw (v_4) -- node [anchor= south] {} (v_5);
			\end{scope} 
			\end{tikzpicture}
		}
		\caption{Graph representation of the network used for estimating the average queue size for flows from $v_1$ to $v_6$, and from $v_7$ to $v_5$.}
		\label{fig:analysis}
	\end{subfigure}
	\begin{subfigure}{0.42\linewidth}
		\centering
		\includegraphics[width=1\linewidth]{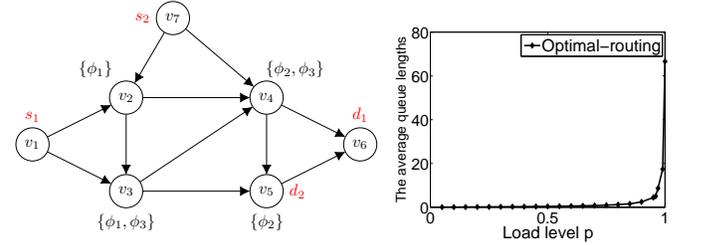}
		\caption{The average queue lengths with different arrival intensities.}
		\label{fig:Result1}
	\end{subfigure}
	\caption{Throughput-optimal routing with SFC constraints.}
	\label{fig:Result1_main}
\end{figure}

Moreover, we integrate our SFC-constrained shortest path with the throughput-optimal routing proposed in \cite{OptimalControl} to have an SFC-constrained throughput-optimal routing algorithm. We conduct simulations for the network shown in Fig. \ref{fig:Result1_main}(\subref{fig:analysis}), where multiple instances of different network functions are available at some vertices. 
We have two flows: flow $f_1$ from $ v_1 $ to $ v_6 $ with an $\text{SFC} =(v_1, \phi_1, \phi_2, v_6)$ and flow $f_2$ from $v_7$ to $ v_5 $ with an $\text{SFC} =(v_7, \phi_1, \phi_3, v_5)$. We assume a unit capacity for each link, i.e., one packet can be sent over each link at each time slot. It can be verified that the max flow rate that can be supported is 2 and 1 for flows $f_1$ and $f_2$, respectively, which is a point at the boundary of the optimal throughput region. We run experiments with Poisson arrivals for flows $f_1$ and $f_2$ with rates $\lambda_1 = 2 p$ and $\lambda_2 = p$, respectively, where $0 \leq p \leq 1$. We run experiments for $10^5$ slots, the first $10^4$ slots are excluded to consider the average queue lengths in the steady state. From the result in Fig. \ref{fig:Result1_main}(\subref{fig:Result1}), we can see that as long as the arrival  rate vector is strictly within the capacity region (i.e., $p<1$), the average total queue length is kept finite under the SFC-constrained throughput optimal routing.

\subsection{VNFs Placement Results}

	\begin{figure}
		\begin{subfigure}{0.56\linewidth}
			\centering	
			\resizebox{\linewidth}{!}{
		\begin{tikzpicture}[transform shape]
		\node[vertex, label=above:\textcolor{red}{source}](v_1) at (0, 0) {$ v_1 $};
		\node[vertex] (v_2) at (2, 1.5) {$ v_2 $};
		\node[vertex](v_3) at (2, 0){$ v_3 $};
		\node[vertex] (v_4) at (2, -1.5) {$ v_4 $};
		\node[vertex](v_5) at (4, 1.5) {$ v_5 $};
		\node[vertex](v_6) at (4, 0) {$ v_6 $};
		\node[vertex](v_7) at (4, -1.5) {$ v_7 $};
		\node[vertex , label=above:\textcolor{red}{destination}](v_8) at (6, 0) {$  v_8 $};
		
		\begin{scope}[every path/.style={->}, every node/.style={inner sep=1pt}]
		
		\draw (v_1) edge [  anchor= south] node {2} (v_2);
		
		\draw (v_1) -- node [anchor=south] {3} (v_3);
		
		\draw (v_1) -- node [anchor= south] {3} (v_4);
		\draw (v_2) -- node [anchor= south] {5} (v_5);
		
		\draw (v_3) -- node [anchor= south] {5} (v_6);
		\draw (v_4) -- node [anchor= south] {3} (v_7);
		\draw (v_4) -- node [anchor= west] {2} (v_3);
		\draw (v_5) -- node [anchor= west] {6} (v_6);
		\draw (v_5) -- node [anchor= south] {3} (v_8);
		\draw (v_6) -- node [anchor= south] {5} (v_8);
		\draw (v_6) -- node [anchor= south] {} (v_5);
		\draw (v_7) -- node [anchor= west] {2} (v_6);
		
		\draw (v_7) -- node [anchor= south] {4} (v_8);

		\end{scope} 
		\end{tikzpicture}
	}
		\caption{A network to study the VNFs placement to achieve the original maximum flow.}
		\label{fig:VNF_placement}
\end{subfigure}
\begin{subfigure}{0.42\linewidth}
	\centering
	\includegraphics[width=1\linewidth]{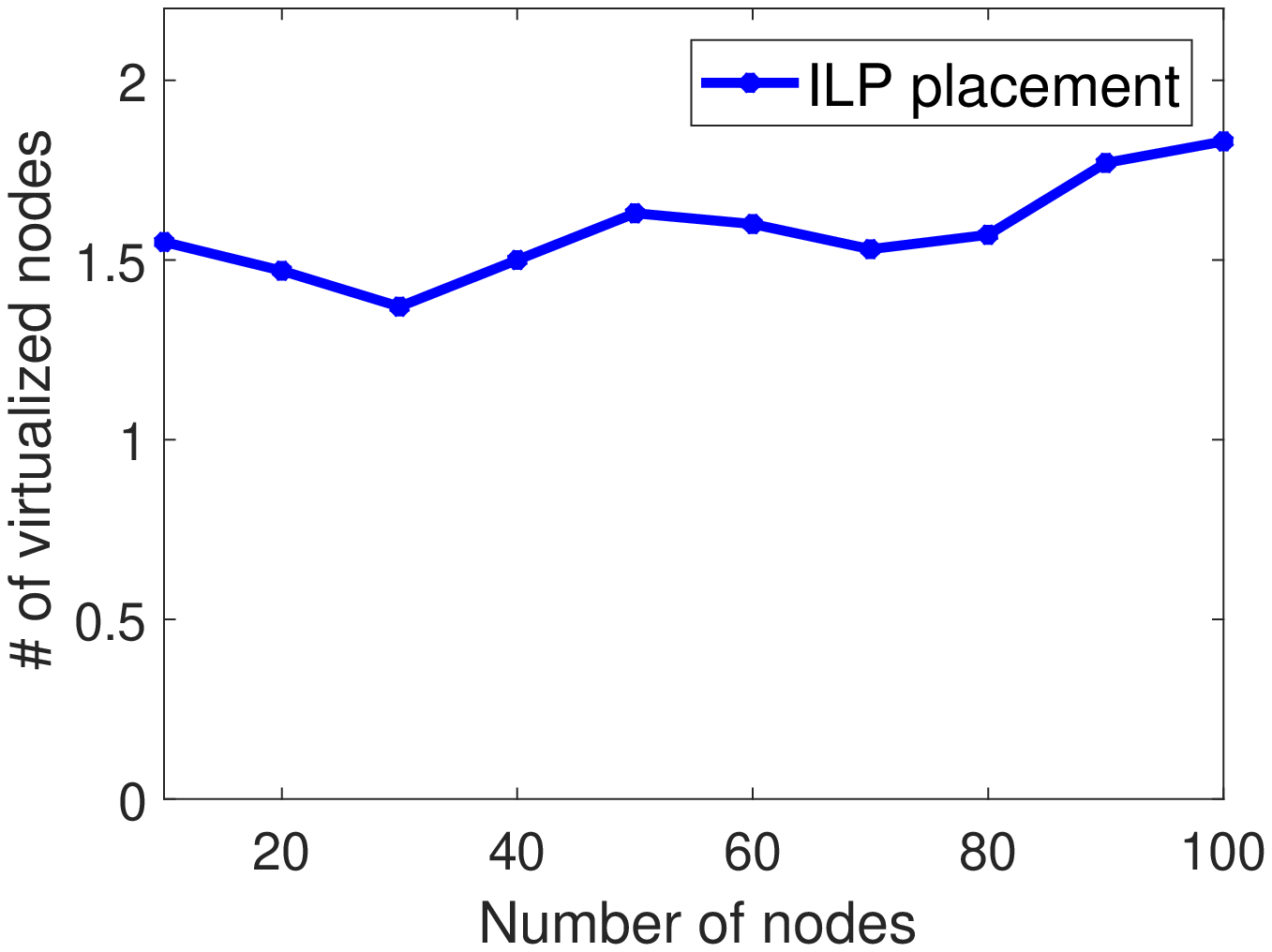}
	\caption{VNFs placement based on the ILP solution.}
	\label{fig:ILP_placement}
\end{subfigure}
\caption{VNFs placement result.}
\label{fig:VNF_placement_main}
\end{figure}

In this subsection, we show some simulation results for the ILP placement. We start by solving the placement for the network shown in Fig. \ref{fig:VNF_placement_main}(\subref{fig:VNF_placement}). We can see that the maximum flow from node $v_1$ to node $v_8$ is 8, which can be through different set of paths. But, we want to minimize the number of nodes to place the required function on them and still achieve the maximum flow of 8. It can be verified that a placement at node $v_6$ only is sufficient, i.e., the maximum flow from node $v_1$ to node $v_8$ through node $v_6$ is 8. Now, we provide a result for more general graphs. We consider random graphs with a different number of nodes from 10 to 100. Each node has an average degree of $|\mathcal{V}|/3$, and the capacity of edges are uniformally distributed between 2 and 10. We repeat each experiment for 20 times and report the average result in Fig. \ref{fig:VNF_placement_main}(\subref{fig:ILP_placement}). Based on the result, we can see that for different graph sizes, we need a small number of nodes to achieve the original maximum flow. That would be an incentive for network operators to introduce VNFs in their networks with low cost without impacting the capacity. The running time of the ILP solution ranges from a few seconds for small graphs to few minutes for large graphs.
\section{Conclusion}
\label{sec:conclusion}
In this paper, we have investigated several issues that arise from Service Functions Chains (SFC) constraints in networks with combined PNFs and VNFs. We solved the SFC-constrained shortest path problem by a transformation of the network graph to a new graph, which ensures an efficient computation of an SFC-constrained shortest path. We also investigated the problem of an SFC-constrained maximum flow problem. We formulated the problem as a fractional multicommodity maximum flow problem and presented a combinatorial solution for a special case. Lastly, we considered VNFs placement from a maximum flow perspective. Our objective is to achieve the original maximum flow in the network while satisfying a given SFC constraint. We showed that the problem is NP-hard. Then, we provided an equivalent ILP formulation, which can be solved in a few minutes for large instances. An interesting problem for future work is to develop a combinatorial algorithm for computing the SFC constrained maximum flow in general. It is also important to find an approximation algorithm for the VNFs placement problem we formulated in this paper.

\bibliographystyle{abbrv}
\bibliography{refrences}

\begin{thebibliography}{10}

\bibitem{Amdocs_whitepaper}
Amdocs.
\newblock {Bringing NFV to Life - Technological and Operational Challenges in
  Implementing NFV}.
\newblock {\em White paper}, 2016.

\bibitem{Arora2012}
S.~Arora, E.~Hazan, and S.~Kale.
\newblock {The Multiplicative Weights Update Method: A Meta-Algorithm and
  Applications}.
\newblock {\em Theory of Computing}, 8(1):121--164, 2012.

\bibitem{Bhamare2016}
D.~Bhamare, R.~Jain, M.~Samaka, and A.~Erbad.
\newblock {A survey on service function chaining}.
\newblock {\em Journal of Network and Computer Applications}, 75:138--155,
  2016.

\bibitem{Bhamare2017}
D.~Bhamare, M.~Samaka, A.~Erbad, R.~Jain, L.~Gupta, and H.~A. Chan.
\newblock {Optimal virtual network function placement in multi-cloud service
  function chaining architecture}.
\newblock {\em Computer Communications}, 102:1--16, 2017.

\bibitem{Cao2014}
Z.~Cao, M.~Kodialam, and T.~Lakshman.
\newblock {Traffic Steering in Software Defined Networks : Planning and Online
  Routing}.
\newblock {\em Proceedings of the 2014 ACM SIGCOMM workshop on Distributed
  cloud computing(DCC '14)}, pages 65--70, 2014.

\bibitem{Charikar2015}
M.~Charikar, Y.~Naamad, J.~Rexford, and X.~Zou.
\newblock {Multi-Commodity Flow with In-Network Processing}.
\newblock {\em DIMACS Networking Workshop}, (1), 2015.

\bibitem{Introduction2012}
M.~Chiosi, D.~Clarke, P.~Willis, A.~Reid, J.~Feger, M.~Bugenhagen, W.~Khan,
  M.~Fargano, C.~Cui, H.~Deng, D.~Telekom, and U.~Michel.
\newblock {Network Functions Virtualisation, An Introduction, Benefits,
  Enablers, Challenges \& Call for Action}.
\newblock {\em in Proc. SDN OpenFlow World Congr., Darmstadt, Germany},
  (1):1--16, 2012.

\bibitem{Dwaraki2016}
A.~Dwaraki and T.~Wolf.
\newblock {Adaptive Service-Chain Routing for Virtual Network Functions in
  Software-Defined Networks}.
\newblock {\em Proceedings of the 2016 ACM Workshop on Hot Topics in
  Middleboxes and Network Function Virtualization}, pages 32--37, 2016.

\bibitem{Feng2016}
H.~Feng, J.~Llorca, A.~M. Tulino, and A.~F. Molisch.
\newblock {Optimal dynamic cloud network control}.
\newblock {\em 2016 IEEE International Conference on Communications, ICC 2016},
  pages 0--6, 2016.

\bibitem{Feng2017}
H.~Feng, J.~Llorca, A.~M. Tulino, D.~Raz, and A.~F. Molisch.
\newblock {Approximation Algorithms for the NFV Service Distribution Problem}.
\newblock {\em IEEE INFOCOM 2017, Atlanta, GA, May 2017}, pages 1--9.

\bibitem{Fleischer1999}
L.~Fleischer.
\newblock {Approximating fractional multicommodity flow independent of the
  number of commodities}.
\newblock {\em 40th Annual Symposium on Foundations of Computer Science (Cat.
  No.99CB37039)}, 13(4):1--16, 1999.

\bibitem{Ford}
L.~R. Ford and D.~R. Fulkerson.
\newblock {Maximal flow through a network}.
\newblock {\em Canadian journal of Mathematics}, 8:399--404, 1956.

\bibitem{Karakostas2008}
G.~Karakostas.
\newblock {Faster approximation schemes for fractional multicommodity flow
  problems}.
\newblock {\em ACM Transactions on Algorithms (TALG)}, 4(1):1--17, 2008.

\bibitem{Li2016a}
Y.~Li, L.~T.~X. Phan, and B.~T. Loo.
\newblock {Network functions virtualization with soft real-time guarantees}.
\newblock {\em Proceedings - IEEE INFOCOM}, 2016-July, 2016.

\bibitem{Ma2017}
W.~Ma, O.~Sandoval, J.~Beltran, D.~Pan, and N.~Pissinou.
\newblock {Traffic Aware Placement of Interdependent NFV Middleboxes}.
\newblock {\em IEEE INFOCOM}, pages 1--9, 2017.

\bibitem{poularakis2017one}
K.~Poularakis, G.~Iosifidis, G.~Smaragdakis, and L.~Tassiulas.
\newblock One step at a time: Optimizing sdn upgrades in isp networks.
\newblock In {\em Proceedings of IEEE INFOCOM}, 2017.

\bibitem{Sang}
Y.~Sang, B.~Ji, G.~R. Gupta, X.~Du, and L.~Ye.
\newblock {Provably Efficient Algorithms for Joint Placement and Allocation of
  Virtual Network Functions}.
\newblock {\em IEEE INFOCOM 2017, Atlanta, GA, May 2017}.

\bibitem{Sherry2012}
J.~Sherry, S.~Ratnasamy, and C.~Sciences.
\newblock {A Survey of Enterprise Middlebox Deployments}.
\newblock (UCB/EECS-2012-24), 2012.

\bibitem{OptimalControl}
A.~Sinha and E.~Modiano.
\newblock {Optimal Control for Generalized Network-Flow Problems}.
\newblock {\em IEEE INFOCOM 2017, Atlanta, GA, May 2017}.

\bibitem{Vardhan2009}
H.~Vardhan, S.~Billenahalli, W.~Huang, M.~Razo, A.~Sivasankaran, L.~Tang,
  P.~Monti, M.~Tacca, and A.~Fumagalli.
\newblock {Finding a simple path with multiple must-include nodes}.
\newblock {\em Proceedings - IEEE MASCOTS}, pages 607--609, 2009.

\end{thebibliography}

\end{document}